\documentclass[article]{IEEEtran}
\usepackage{wrapfig}
\usepackage{amsfonts}
\usepackage{amssymb,epic,eepic}
\usepackage{amsmath}
\usepackage{dsfont}
\usepackage{dcolumn}
\usepackage{bm}
\usepackage{bbm}
\usepackage{verbatim}
\usepackage{color}
\usepackage{amsthm}
\usepackage{tikz}
\usepackage{cite}
\usepackage{subfig,float}
\usetikzlibrary{positioning}
\usepackage{MnSymbol,wasysym}
\usetikzlibrary{arrows,positioning,decorations.pathmorphing,decorations.markings,shapes,fadings,patterns}
\usepackage{pstool}
\usepackage{psfrag}
\usepackage{import}
\usepackage{varwidth}
\usepackage{graphicx}

\newcommand{\cX}{{\mathcal X}}

\newcommand{\cY}{{\mathcal Y}}

\newcommand{\nn}{{\mathbb N}}

\newcommand{\eps}{{\varepsilon}}        %%%%%%%%%%%%%%%%%%%%%%%%%%%%%%%%%%%

\newcommand{\by}{{\mathbf y}}

\newcommand{\cP}{\mathcal P}

\newtheorem{theorem}{Theorem}

\newtheorem{algorithm}{Algorithm}

\newtheorem{definition}[theorem]{Definition}
\newtheorem{example}[theorem]{Example}

\newtheorem{lemma}[theorem]{Lemma}

\newtheorem{remark}[theorem]{Remark}
\newtheorem{solution[theorem]}{Solution}

\newcommand{\supp}{\mathrm{supp}}
\newcommand{\eins}{{\mathbbm{1}}}
\makeatletter
\newcommand*{\rom}[1]{\expandafter\@slowromancap\romannumeral #1@}
\makeatother
\DeclareMathOperator{\sgn}{sgn}

\DeclareRobustCommand*{\IEEEauthorrefmark}[1]{\raisebox{0pt}[0pt][0pt]{\textsuperscript{\footnotesize #1}}}
\begin{document}

\pagestyle{plain}
\thispagestyle{empty}

\title{Information Theoretic Principles \protect\\ of Universal Discrete Denoising}

\author{\IEEEauthorblockN{Janis N\"otzel\IEEEauthorrefmark{$\ast$,1,2} and
                          Andreas Winter\IEEEauthorrefmark{$\times$,1,3}}

\IEEEauthorblockA{\IEEEauthorrefmark{1}%
 Departament de F\'{\i}sica: Grup d'Informaci\'{o} Qu\`{a}ntica,\\
 Universitat Aut\`{o}noma de Barcelona,
 ES-08193 Bellaterra (Barcelona), Spain.}

\IEEEauthorblockA{\IEEEauthorrefmark{2}%
 Theoretische Nachrichtentechnik, Technische Universit\"at Dresden,
 D-01187 Dresden, Germany.}

\IEEEauthorblockA{\IEEEauthorrefmark{3}%
 ICREA---Instituci\'{o} Catalana de Recerca i Estudis Avan\c{c}ats, Pg.~Lluis Companys, 23,
 ES-08010 Barcelona, Spain.}

\IEEEauthorblockA{Email: \IEEEauthorrefmark{$\ast$}janis.noetzel@tu-dresden.de,
                         \IEEEauthorrefmark{$\times$}andreas.winter@uab.cat}

\IEEEauthorblockA{\smallskip(23 May 2017)}
}

\maketitle
	
\begin{abstract}
Today, the internet makes tremendous amounts of data widely available.
Often, the same information is behind multiple different available data sets.
This lends growing importance to latent variable models that try to learn the
hidden information from the available \emph{imperfect}€ versions. For example,
social media platforms can contain an abundance of pictures of the same person
or object, yet all of which are taken from different perspectives.
In a simplified scenario, one may consider pictures taken from the same
perspective, which are distorted by noise. This latter application allows for
a rigorous mathematical treatment, which is the content of this contribution.
We apply a recently developed method of dependent component analysis to image
denoising when multiple distorted copies of one and the same image are available,
each being corrupted by a different and unknown noise process.

In a simplified scenario, we assume that the distorted image is corrupted
by noise that acts independently on each pixel. We answer completely the
question of how to perform optimal denoising, when at least three distorted
copies are available: First we define optimality of an algorithm in the presented
scenario, and then we describe an aymptotically optimal universal discrete denoising
algorithm (UDDA).
In the case of binary data and binary symmetric noise, we develop a
simplified variant of the algorithm, dubbed BUDDA, which we prove to
attain universal denoising uniformly.
\end{abstract}


%%%%%%%%%%%%%%%%%%%%%%%%%%%%%%%%%%%%%%%%%%%%%%%%%%%%%%%%%%%%%%%%%%%%%%%%%%%%%%%%%%%%
%%%%%%%%%%%%%%%%%%%%%%%%%%%%%%%%%%%%%%%%%%%%%%%%%%%%%%%%%%%%%%%%%%%%%%%%%%%%%%%%%%%%
\begin{section}{Introduction}
\label{sec:intro}

Consider the following, idealised situation of inference from noisy data:
In the past, a number $K$ of pictures has been taken of an object.
Each picture consists of a number $n$ of pixels, and each pixel takes
values in a finite alphabet $\cX$ of colours, $|\cX|=L$. We would like to
recover from these images a ``true'' image, of which we assume all $K$
that we have at our disposal are noisy copies.

To be precise, we assume that each of the pictures is the output
of a memoryless (i.i.d.) noise process $W$, mapping an element
$x\in\cX$ to another element $y\in\cX$ with probability $w(y|x)$.
%(where we assume $|\cY|\geq|\cX|$ throughout).
%For those cases where $|\cX|=|\cY|$ one may w.l.o.g.~assume $\cX=\cY$ and
If we knew $W$ as well as the frequency $p(x)$ of the colours $x$ in
the original, we could choose some permutation $\Phi:\cX\to\cX$
as the denoising rule.
If one is agnostic about the colour of each pixel and applies
any such denoising to a long enough string on a pixel-by-pixel basis, it will produce
a recovered image that, with high probability, deviates from the original
in a fraction of pixels asymptotically equal to
\begin{equation}\begin{split}
  1-\max_{\tau} &\sum_{x\in\cX}p(\tau(x))w(\Phi^{-1}(x)|\tau(x)) \\
                &\quad
                 = 1 - \max_{\tau}\sum_{x\in\cX}p(x)w(\tau(x)|x),
  \label{eqn:colour-agnostic-solution}
\end{split}\end{equation}
where the maximum is over all permutations $\tau:\cX\to\cX$.
Given that we are more interested in the
structure of the image, and not in the question whether it is, say,
black-and-white or white-and-black, this is a satisfactory solution.

However, if the number of possible outputs of the noise process grows, such simple
solutions turn out to be not optimal anymore, because more information is available
in the string that a letter-by-letter algorithm does not see.
Consider the following example, where a channel maps the black and white pixels
represented by the alphabet $\{1,2\}$ to three different colours $\{1,2,3\}$.
Let the probabilistic law of this process be given by $w(1|1)=w(3|1)=1/2$ and $w(2|2)=1$.
A simple denoising algorithm may map the elements of the set $\{1,2\}$ to $1$
and those in $\{3\}$ to $2$. If the input bits are equally distributed, the
probability that this algorithm produces an error is lower bounded by $1/4$
-- even if one is agnostic as to an exchange $1\leftrightarrow 2$ at the input.
Instead, knowing the channel lets one perform perfect denoising,
i.e.~with zero error.

The problem is thus to know the channel, which it turns out can be
achieved by looking at several independent noisy copies instead of only one.
We now give another example about large classes of channels which allow
us to achieve a nontrivial inference about the original
by a generic algorithmic strategy.

\begin{example}
Let the binary symmetric channels (BSC)
$W_1$, $W_2$, $W_3$ be given, each with a respective probability
$w_i$ of transmitting the input bit correctly, and the channel $W$ from $\{0,1\}$
to $\{0,1\}^3$ is the product, giving the same input to all:
$W(y_1y_2y_3|x) = \prod_{j=1}^3 w_j(y_j|x)$. Then, a simple way of guessing the
input bit is to use the majority rule, i.e.~the output is mapped to $0$ iff
$N(0|y_1y_2y_3)>N(1|y_1y_2y_3)$. Assume that the input string has equal number of
$0$ and $1$.
For BSC parameters $w_1=0.1$, $w_2=0.45$ and $w_3=0.9$, this gives a probability
%$p_\mathrm{correct}$
of correct decoding that is bounded as $4/10\leq p_\mathrm{correct}\leq 5/10$.

Thus, even taking into account our agnostic view regarding the labelling of the input
alphabet, the probability of guessing the right input up to a permutation is upper bounded
by $6/10$.

Now consider a genie-aided observer that is given BSC parameters.
Such an observer would come to the conclusion that the second channel, $W_2$, actually
conveys little information to him, and could decide to not consider its output at all.
The genie-aided observer could then come up with the following rule for guessing
the input bit: The outputs in the set $\{00,01\}$ get mapped to $0$,
while the outputs in the set $\{11,10\}$ get mapped to $1$.
Then, the probability of correctly guessing the input bit would rise to $9/10$.
\end{example}

Note that the permutation of
the colours is an unavoidable limit to any recovery procedure in which we
do not want to make any assumptions about the channels and the colour
distribution $p$.

%In this particular case however, dependent component analysis offers a generic way
%to replace the genie and obtain channel information for a large class of channels
%simply by application of one algorithm.

\medskip
Based on these considerations, we propose
the general mathematical formulation of our problem is as follows.
As before we assume the existence of one ``original'' version of the picture, which
is represented by a string $x^n\in\cX^n$ where $n\in\mathbb N$. A number $K$ of
copies of the image is given, where each copy is modelled by a string
$\mathbf y_j=(y_{1j},\ldots,y_{nj})$, and $j$ ranges from $1$ to $K$.
These images are assumed to be influenced by an i.i.d.~noise process, such that the
probability of receiving the matrix $\by=(y_{ij})_{i,j=1}^{n,K}$ of distorted images
is given by
\begin{align}
  \label{eqn:probabilistic-law}
  \mathbb P(\by|x^n) = \prod_{i,j=1}^{n,K}w_j(y_{ij}|x_i).
\end{align}
The probabilistic laws $w_1,\ldots,w_K$ crucially are not known.
The question is: What is the optimal way of recovering $x^n$ from the output $\by$?
Before we formalize this problem, we give an overview over the structure of the paper and then of the related work.

\end{section}

%%%%%%%%%%%%%%%%%%%%%%%%%%%%%%%%%%%%%%%%%%%%
\begin{section}{\label{sec:related-work}Related work}
The earlier work in the area of image denoising has focussed mostly on situations where the probabilistic law governing the noise process is (partially) known. Among the work in that area are publications like \cite{wosvw-DUDE,ordentlich-seroussi-weinberger,gemelos-sigurjonsson-weissman-1,gemelos-sigurjonsson-weissman-2}. In particular, the work \cite{wosvw-DUDE} laid the information theoretic foundation of image denoising when the channel is known. It has been extended to cover
analysis of grayscale images, thereby dealing with the problem incurred by a large alphabet, here \cite{morsw-iDUDE}. The DUDE framework
has also been extended to cover continuous alphabets \cite{morsw-continuous-DUDE} and in other directions \cite{ovw-twice-universal}.
An extension to cover noise models with memory has been developed in \cite{gy-DUDE-with-memory}. Optimality of the DUDE has been further
analyzed in \cite{vo-lower-limits-DUDE}, where it was compared to an omniscient denoiser, that is tuned to the transmitted noiseless
sequence. An algorithm for general images has e.g. been demonstrated here \cite{cao-luo-yang}.

Cases where the probabilistic law behind the noise process is known partially were treated in
\cite{gemelos-sigurjonsson-weissman-1,gemelos-sigurjonsson-weissman-2}.
A tour of modern image filtering with focus on denoising can be found in the work \cite{milanfar}. The feasibility of the task of multi
copy image denoising in a very particular scenario that is similar to the one treated here has been demonstrated
\cite{shih-kwatra-chinen-fang-ioffe}. However, that work assumes Gaussian noise, whereas our work is able to handle any form of noise within the category of finite-alphabet memoryless noise processes. Sparsity has been exploited for coloured images in \cite{mairal-elad-sapiro}, universal compressed sensing in \cite{jalali-poor}. A new method for image denoising was presented first here \cite{bcm-NL-means} and then further developed in
\cite{buades-coll-morel}. A scenario that is related to the one treated here has been investigated in the recent work \cite{ns16}, but
without application to a specific problem. The focus of the work \cite{ns16} was to demonstrate that certain hidden variable models have a one to one correspondence to a particular subset of probability distributions. The results presented in \cite{ns16} are closely related to the famous work \cite{Kruskal76} of Kruskal, which forms one theoretical cornerstone of the PARAFAC model (see the review article \cite{acar-yener-2009} to get an impression of methods in higher order data analysis). The connections between PARAFAC, multivariate statistical analysis \cite{rencher-cristensen-2012} and more information-theoretic models has also been used in the recent work \cite{AGHKT14}, where the authors basically elaborated on methods to efficiently perform step $3$ of our algorithm in the restricted scenario where all channels are identical: $w_1(y|x)=\ldots=w_K(y|x)$ for all $y\in\cY$ and $x\in\cX$. The interested reader can find details regarding algorithms for fitting the PARAFAC model in \cite{tomasi-bro2006}. Further connections to deep learning, current implementations and other areas of application are to be found in \cite{ko-pa-pa}. Our Theorems \ref{theorem:basis-for-UDDA} and \ref{theorem:main-BSC} provide an improvement over all the existing algorithms in a restricted but important scenario.
Despite some similarities, the problem treated here is to be distinguished from the CEO problem \cite{prabhakaran-tse-ramachandran2004,berger-zhang1994}, which is also receiving attention recently \cite{kipnis-rini-goldsmith} but is focussed on rate constraints in the process of data combining from the multiple sources. In addition, noise levels are assumed to be known in the CEO problem.
\end{section}

%%%%%%%%%%%%%%%%%%%%%%%%%%%%%%%%%%%%%%%%%%%%%%%%%%
%%%%%%%%%%%%%%%%%%%%%%%%%%%%%%%%%%%%%%%%%%%%%%%%%%
\begin{section}{Definitions}
\begin{subsection}{Notation\label{subsec:notation}}
The set of permutations of the elements of a finite set $\cX$ is denoted $S_\cX$.
Given two finite sets $\cX,\cY$, their product is $\cX\times\cY:=\{(x,y):x\in\cX,y\in\cY\}$.
For any natural number $n$, $\cX^n$ is the $n$-fold product of $\cX$ with itself.
For a string $x^n\in\cX^n$, and an element $x'\in\cX$, let
$N(x'|x^n) := |\{i: x_i=x'\}|$ be the number of times $x'$ occurs in the string.
The probability distribution $\overline{N}(x'|x^n) := \frac1n N(x'|x^n)$
is called the \emph{type} of $x^n$.
The set $\mathcal P(\cX)$ of probability distributions on $\cX$ is identified with
the subset
\begin{align}
  \left\{p\in\mathbb R^{|\cX|} : p=\sum_{x\in\cX}p(x)e_x\right\}
\end{align}
of $\mathbb{R}^{|\cX|}$, where $\{e_x\}_{x\in\cX}$ is the standard orthonormal basis of $\mathbb R^{|\cX|}$. The support of
$p\in\cP(\cX)$ is $\supp(p):=\{x\in\cX:p(x)>0\}$.
We define two important subsets of $\cP(\cX)$, for a conventional
numbering $\cX = \{x_1,x_2,\ldots,x_L\}$:
\begin{align}
  \cP_>(\cX)          &:= \left\{p\in\cP(\cX) : p(x)>0\ \forall\ x\in\cX\right\}, \\
  \cP^\downarrow(\cX) &:= \left\{p\in\cP(\cX) : p(x_1) \geq p(x_2) \geq \ldots \geq p(x_L) > 0 \right\}.
\end{align}
A subset of elements of $\cP(\cX)$ is the set of its extreme points,
the Dirac measures:
for $x\in\cX$, $\delta_x\in\cP(\cX)$ is defined through $\delta_x(x')=\delta(x,x')$,
where $\delta(\cdot,\cdot)$ is the usual Kronecker delta symbol. The indicator
functions are defined for subsets $S\subset\cX$ via $\eins_{S}(x)=1$ if
$x\in S$ and $\eins_{\cX}(x)=0$ otherwise.

We will have to measure how different two elements $p,p'\in\cP(\cX)$ are. This may be
done using (for any $\alpha\geq1$) the $\alpha$-norms
$\|p-p'\|_\alpha:= \sqrt[\alpha]{\sum_{x\in\cX}|p(x)-p'(x)|^\alpha}$.
If $\alpha=1$, this reduces, up to factor of $2$, to the total variational distance,
which we will use in the following while dropping the subscript.

The noise that complicates the task of estimating $p$ is modelled by matrices $W$
of conditional probability distributions $(w(y|x))_{x\in\cX,y\in\mathbf Y}$ whose
entries are numbers in the interval $[0,1]$ satisfying, for all $x\in\cX$,
$\sum_{y\in\cY}w(y|x)=1$; such objects are also called stochastic matrices.
To given finite sets $\cX$ and $\cY$, the set of all such matrices is denoted
$C(\cX,\cY)$ in what follows. These stochastic matrices are, using standard
terminology of Shannon information theory, equivalently called a ``channel''.
In this work, we restrict attention to those stochastic matrices that are invertible,
denoting their set $C_0(\cX,\cY)$.
A special channel is the ``diagonal'' $\Delta\in C(\cX,\cX^K)$, defined
as $\Delta(\delta_x):=\delta_x^{\otimes K}$. For a probability distribution
$p\in\cP(\cX)$, we write $p^{(K)} := \Delta(p)$, which assigns probability
$p(x)$ to $xx\ldots x$, and $0$ to all other strings.
If two channels $W_1$ and $W_2$ act independently in parallel (as in
\eqref{eqn:probabilistic-law} with $K=2$ and $n=1$), one writes $W_1\otimes W_2$
for this channel, and more generally for $K\geq 2$.
We will abbreviate the $K$-tuple $(W_1,\ldots,W_K)$ as $\mathbf W$, if
$K$ is known from the context.
A \emph{dependent component system} is simply a collection
$(p,\mathbf{W}) = (p,W_1,\ldots,W_K)$ consisting of a distribution $p$ on $\cX$
and channels $W \in C(\cX,\cY)$.
The distance between two dependent component systems is defined as
\begin{align}
  \|(r,\mathbf V)-(s,\mathbf W)\|_{DCS} := \|r-s\|_1+\sum_{j=1}^K\|V_j-W_j\|_{FB},
\end{align}
where $\|V_j-W_j\|_{FB} := \max_x \sum_{y} |v_j(y|x)-w_j(y|x)|$.

Finally, we quantify the distance between an image (or string of letters)
and another such image (or string of letters).

\begin{definition}[Distortion measure]
\label{defn:distortion-measure}
Any function $d:\cX\times\cX\to\mathbb R_{\geq 0}$ is said to be a
\emph{distortion measure}.
Clearly, this definition includes as a subclass the distortion measures that
are at the same time a metric. The special case $d(x,x'):=\delta(x,x')$
of the Kronecker delta
is known as the Hamming metric. It is understood that any distortion measure
$d$ naturally extends to a distortion measure on $\cX^n\times\cX^n$ via
\begin{align}
\label{eqn:distance-measure-on-n-fold-alphabet}
  d(x^n,y^n) := \frac{1}{n}\sum_{i=1}^n d(x_i,y_i).
\end{align}
\end{definition}
\end{subsection}

%%%%%%%%%%%%%%%%%%%%%%%%%%%%%%%%%%%%%%%%%%%%%%%%%%%%%%%%%%%%%%%%%%%%%%%%%%%%%%%%
\begin{subsection}{Definition of Resolution}
In the following, we assume a distortion measure $d$ on $\cX\times\cX$.
With respect to such a measure, the performance of a denoising algorithm
can be defined. In the special case treated here, we are interested in a
form of universal performance that does not depend on any knowledge regarding
the specific noise parameters.
Our way of defining universality is in the spirit of Shannon information theory.
Such a definition brings with it two important features: First, it is an
asymptotic definition, and secondly it is a definition that is agnostic as
to the specific meaning of a pixel color. For example, we do not make any
distinction between a black-and-white picture and a version of that picture
where black and white have been interchanged throughout. This motivates the
following definition.

\begin{definition}[Universal algorithm]
\label{defn:weak-performance-of-an-algorithm}
\label{defn:performance-of-an-algorithm}
Let $\mathbf W\in C(\cX,\cX)^K$ and $p\in\mathcal{P}(\cX)$ be given.
A sequence $\mathcal C=(C_n)_{n\in\nn}$ of algorithms $C_n:(\cX^K)^n\to\cX^n$ is
said to achieve \emph{resolution $D\geq 0$ on $\mathbf W$ with respect to $p$},
if for every $\eps>0$ there is an $N=N(p,\mathbf W,\eps)$ such that for
all sequences $(x^n)_{n\in\nn}$ of elements of $\cX^n$ with
limit $\lim_{n\to\infty} \overline{N}(\cdot|x^n)=p$ we have
\begin{align}
  \forall n\geq N(p,\mathbf{W},\eps)\quad
  \mathbb P\left( \min_\tau d\bigl(C_n(\cdot),\tau^{\otimes n}(x^n) \bigr) > D \right) < \eps.
\end{align}
The number
\begin{align}
\mathcal R(p,\mathbf W)
  := \inf\left\{D :\begin{array}{l} \exists\,\mathcal C \text{ achieving distortion}\\
                                    \phantom{===:} D \text{ with respect to } p
                   \end{array}\right\}
\end{align}
is the \emph{resolution of $\mathbf W$ with respect to $p$}.
A sequence of algorithms is \emph{universal} if it achieves
resolution $\mathcal R(p,\mathbf W)$ for all $p$ and $\mathbf W$.

A sequence of algorithms is \emph{uniformly universal} on
$\mathcal S\subset C(\mathcal X,\mathcal X)$ if for all
$\mathbf W\in\mathcal S^K$ the number
$N=N(\mathbf{W},\eps)=N(p,\mathbf{W},\epsilon)$ can be chosen
independently of $p$. If such a sequence exists, we say that the
\emph{resolution $D$ is achieved uniformly}. Then,
\begin{align}
\overline{\mathcal{R}}(p,\mathbf W)
  := \inf\left\{D :\begin{array}{l} \exists\,\mathcal C \text{ uniformly achieving}\\
                                    \phantom{===:} \text{distortion } D \text{ w.r.t. } p
                   \end{array}\right\}
\end{align}
is called the \emph{uniform resolution of $\mathbf W$}.
Clearly, $\mathcal{R}(p,\mathbf W) \leq \overline{\mathcal{R}}(p,\mathbf W)$.
\end{definition}
Note that, if a sequence of uniformly universal algorithms exists, this implies that these algorithms deliver provably optimal performance
for \emph{every} possible input sequence $x^n$, if only $n\geq N(\mathbf W,\eps)$ is large enough.

We stress that the definition of resolution does in general depend on
$p$, as can be seen from the following example.

\begin{example}
Let $W\in\mathcal C(\{0,1,2\},\{0,1\})$ be defined by $w(0|0)=w(0|1)=1$
and $w(1|2)=1$ and $d$ be the hamming distortion. Then, clearly,
every two $p,\hat p\in\cP(\{0,1,2\})$ satisfying $p(2)=\hat p(2)=0$ have
the property $W(p)=W(\hat p)$. Thus, no two sequences $x^n,\hat
x^n\in\{0,1,2\}^n$ having $N(2|x^n)=N(2|\hat x^n)=0$ can be distinguished
from each other after having passed through the channel.

Let $C_n$ be any decoder, assuming without loss of
generality that $C_n(1^n)=1^n$. Consider $x^n=0^{n/2}1^{n/2}$; it follows
that $\min_\tau d(C_n(\cdot),\tau^{\otimes n}(x^n))=\frac12$ with probability $1$.

The same decoder applied to $\hat x^n=1^n$ or $\tilde x^n=0^n$ obviously
has $\min_\tau d(C_n(\cdot),\tau^{\otimes n}(\hat x^n))=0$ and
$\min_\tau d(C_n(\cdot),\tau^{\otimes n}(\tilde x^n))=0$, both once more with probability $1$.

This demonstrates the dependence of $\mathcal R(p,\mathbf W)$ on $p$ for
simple cases where $K=1$ and the channel matrix is non-invertible.
\end{example}

\begin{remark}
Note that our definition of universality explicitly forbids an
algorithm to know the probabilistic laws according to which the original
image gets distorted.
In that sense, it respects a stronger definition of universality as
for instance the DUDE algorithm \cite{wosvw-DUDE}: the DUDE algorithm
has to know the probabilistic law, which is assumed to be represented by
an invertible stochastic matrix. We will see in
Theorem \ref{theorem:main} in which situations an algorithm obeying our
stronger definition can be expected to deliver provably optimal
performance, and what its drawbacks in terms of performance are.
\end{remark}

%We need a second, weaker definition of universality:
%
% ICH HABE DIE FOLGENDE DEFINITION GEAENDERT, DA ES SICH NICHT UM EINE
% "PROBABILITY" HANDELT (AW).
%
\begin{definition}[Minimal clairvoyant ambiguous distortion]
Let $W\in C(\cX,\cY)$ be a channel and $p\in\mathcal P(\cX)$ a probability distribution.
Let $d:\cX\times\cX\to\mathbb R$ be a distortion measure.
The number
\begin{align}
  d_{\mathrm{MCA}}(p,W) := \min_{\tau\in S_\cX,\mathcal{T}} \sum_{x,x'\in\cX} p(x)d(x,x') w(T_{\tau(x')}|x),
\end{align}
is called the \emph{minimal clairvoyant ambiguous (MCA) distortion},
where each $\mathcal T=\{T_x\}_{x\in\cX}$ is partition of $\cY$,
i.e.~$T_x\cap T_{x'}=\emptyset$ for all $x\neq x'$ and $\bigcup_xT_x=\cY$.
It is the smallest expected distortion obtainable if the original $x^n$
is distributed according to $p^{\otimes n}$ and the dependent component
system $(p,W)$ is known to the decoder.
We call a minimizing partition $\mathcal T$ \emph{MCA decoder}.

For each collection $\mathbf W=(W_1,\ldots,W_K)$ of channels we define
\begin{align}
  d_{\mathrm{MCA}}(p,\mathbf W)
    := d_{\mathrm{MCA}}\left(p,\left(\bigotimes_{j=1}^K W_j\right)\circ\Delta\right),
\end{align}
with the ``diagonal'' channel $\Delta\in C(\cX,\cX^K)$.
\end{definition}
\end{subsection}
\end{section}

%%%%%%%%%%%%%%%%%%%%%%%%%%%%%%%%%%%%%%%%%%%%%%%%%%%%%%%%%%%%%%%%%%%%%%%%%%%%%%%%%%%%%%%%%%%%%%
%%%%%%%%%%%%%%%%%%%%%%%%%%%%%%%%%%%%%%%%%%%%%%%%%%%%%%%%%%%%%%%%%%%%%%%%%%%%%%%%%%%%%%%%%%%%%5
\begin{section}{Main results}
\label{sec:universal-algorithms-for-optimal-recovery}
In this section, we state our main result regarding universal
denoising and present two algorithms.
The first one is universal, while the second one is even uniformly universal,
but unfortunately only if restricted to dependent component systems with
binary symmetric channels.

\begin{subsection}{The universal discrete denoising algorithm}
\begin{theorem}
\label{theorem:main}
Let $K\geq3$ and $d$ be a distortion measure on $\cX\times\cX$.
Let $\mathbf{W}=(W_1,\ldots,W_K)\in C(\cX,\cX)^K$ be such that their
associated stochastic matrices are invertible, and $p\in\mathcal{P}(\cX)$.
Then, the resolution of $\mathbf{W}$ with respect to $p$ is
\begin{align}
  \mathcal R(p,\mathbf{W}) = d_{\mathrm{MCA}}(p,\mathbf W).
\end{align}
There exists an algorithm which is universal on the set $C_0(\cX,\cX)^K$.
Below, this universal discrete denoising algorithm (UDDA) is described.
\end{theorem}

\begin{algorithm}[UDDA: Universal discrete denoising algorithm]
%\label{algo:UDDA}
The following steps produce the desired estimate $\hat x^n$ for $x^n$:
\begin{enumerate}
\item Input is $(y_{ij})_{i,j=1}^{n,K}$.
\item Calculate $\hat q(\cdot):=\tfrac{1}{n}N(\cdot|(\by_1,\ldots,\by_K)) \in \cP(\cX^K)$.
\item Run the DCA algorithm \cite{ns16} (see Definition~\ref{def:DCA-algo} below) with input $\hat q$
      to find stochastic matrices $V_1,\ldots,V_K\in C(\cX,\cX)$ and $r\in\mathcal P(\cX)$.
\item Find the minimal clairvoyant ambiguous recovery $\mathcal T$ for $(r,V_1,\ldots,V_K)$.
\item For each $i\in[n]$, apply the MCA decoder to the sequence $(y_{i1},\ldots,y_{iK})$.
      This generates a sequence $\hat x^n = \hat x_1 \ldots \hat x_n$.
\item Output $\hat x^n$.
\end{enumerate}
\end{algorithm}

\begin{remark}
An obvious modification of this algorithm is to deviate from the above definition in
step 4, and instead run the DUDE algorithm from Ref.~\cite{wosvw-DUDE} or one of its modifications.
The results in \cite{gemelos-sigurjonsson-weissman} imply that optimality holds also
in that case, since the deviation of the estimated tuple $(r,V_1,\ldots,V_K)$
from the original $(p,W_1,\ldots,W_K)$ goes to $0$ as $n\to\infty$.
\end{remark}

\begin{definition}
\label{def:DCA-algo}
The DCA algorithm from \cite{ns16} is defined by the following simple rule:
\begin{align}
  \hat q \longmapsto
   \arg\min\left\{\bigl\|(\bigotimes_{i=1}^KV_i)r^{(K)}-\hat q\bigr\|
                   : \begin{array}{l} r\in\mathcal P^\downarrow(\cX),\\
                                      V_i\in C(\cX,\cX)
                     \end{array}\right\}.
\end{align}
\end{definition}

The reason that this simple algorithm yields correct solutions is based on
the following observation, which is a reformulation of
\cite[Thm.~1]{ns16} and stated without proof.

\begin{theorem}
\label{theorem:extended-uniqueness-of-solution}
Let $K\geq3$, $p\in\cP(\cX)$ and let $W_1,\ldots,W_K\in C_0(\cX,\cX)$
be invertible channels. There are exactly $|\cX|!$ tuples $(V_1,\ldots,V_K,p')$
of channels $V_1,\ldots,V_K\in C(\cX,\cX)$ and probability distributions $p'\in \cP(\cX)$
satisfying the equation
\begin{align}
\label{eqn:invertibility-L-L}
\sum_{x\in\cX}\prod_{i=1}^Kw_i(y_i|x)p(x)=\sum_{x\in\cX}\prod_{i=1}^Kv_i(y_i|x)p'(x)
\end{align}
simultaneously for all $y_1,\ldots,y_K\in\cX$.
These are as follows: For every permutations $\tau:\cX\to\cX$, the matrices
$V_i=W_i\circ\tau^{-1}$ and the probability distribution $p'=\tau(p)$
solve (\ref{eqn:invertibility-L-L}), and these are the only solutions.
As a consequence, the function $\Theta:\cP(\cX)\times\mathcal W(\cX,\cX)^{K}\to\cP(\cX^K)$ defined by
\begin{align}
\label{eqn:def-of-theta}
\Theta[\left(p,T_1,\ldots,T_K\right)]:=\left(\bigotimes_{i=1}^KT_i\right)\left(\sum_{i=1}^Lp(i)\delta_i^{\otimes K}\right)
\end{align}
is invertible if restricted to the correct subset:
there exists a function $\Theta':\mathrm{ran}(\Theta)\to\cP^\downarrow(\cX)$ such that
\begin{align}
\label{eqn:inversion-property-of-Theta'}
\Theta'(\Theta[(p,T_1,\ldots,T_K)])=p,
\end{align}
for all $p\in\cP^\downarrow(\cX)\cap\cP_>(\cX)$ and $(T_1,\ldots,T_K) \in C_0(\cX,\cX)^K$.
\end{theorem}
\end{subsection}

%%%%%%%%%%%%%%%%%%%%%%%%%%%%%%%%%%%%%%%%%5
%%%%%%%%%%%%%%%%%%%%%%%%%%%%%%%%%%%%%%%%%%
\begin{subsection}{A uniformly universal algorithm}
\label{sec:discussion-of-uniform-universality}
An important distinction that we made in the definitions is
that between a universal and a uniformly universal algorithm.
We will now explain this difference based on the restricted class
of binary symmetric channels (BSCs). Then, we present an algorithm
that is uniformly universal if restricted to dependent component s
ystems allowing only BSCs.

For this class of channels we know the following from \cite[Theorem 3]{na16}:
\begin{theorem}
\label{theorem:DCA_BSC_K_2}
Let $\cX=\cY_1=\cY_2\{0,1\}$ and $K=2$. Let $A,C\in C_0(\cX,\cY_1)$ and $B,D\in C_0(\cX,\cY_2)$be BSCs and $p,r\in\mathcal{P}_>(\cX)$.
The four equations
 \begin{align}\label{eqn:DCA-for-BSC-and-K=2}
 	\sum_{x\in\cX}a(y_1|x)b(y_2|x)p(x)=\sum_{x\in\cX}c(y_1|x)d(y_2|x)r(x)
\end{align}
(one for each choice of $y_1,y_2\in\{0,1\}$) have exactly the following two
families of solutions:
either $A=C$, $B=D$ and $p=r$,
or $A = C\circ\mathbb{F}$, $B=D\circ\mathbb{F}$, $p = \mathbb{F}(r)$.
\end{theorem}

Since our goal now is to understand the intricacies of uniform universality,
let us for the moment assume that UDDA is uniformly universal. Taking into
account the content of Theorem \ref{theorem:DCA_BSC_K_2}, this might lead one to
conclude that it is already universal for the set of invertible BSCs
and $K=2$. That this is not the case can be seen as follows:
Let $p\in\cX$ be such that $p(0)=p(1)=1/2$; this is the input distribution
explicitly excluded in Theorem \ref{theorem:DCA_BSC_K_2}.
Then (see the appendix for a proof) we get the output distribution
defined by
\begin{align}\label{eqn:BSC-problemfor-K=2}
q(x_1,x_2)=p(x_1)c(x_2|x)\qquad\forall x_1,x_2\in\{0,1\},
\end{align}
where the number $c(0|0)$ is given by $c(0|0)=1-a(0|0)-b(0|0)+2\cdot a(0|0)\cdot b(0|0)$.
Obviously, there is a large number of channels $A$ and $B$ leading to the
same channel $C$, so that UDDA cannot choose the correct one uniquely.
This is also reflected in Fig.~\ref{fig:tetra}.

\begin{figure}[ht]
\begin{flushleft}
\begin{tikzpicture}
\node at (-2,0) {\includegraphics[scale=0.34]{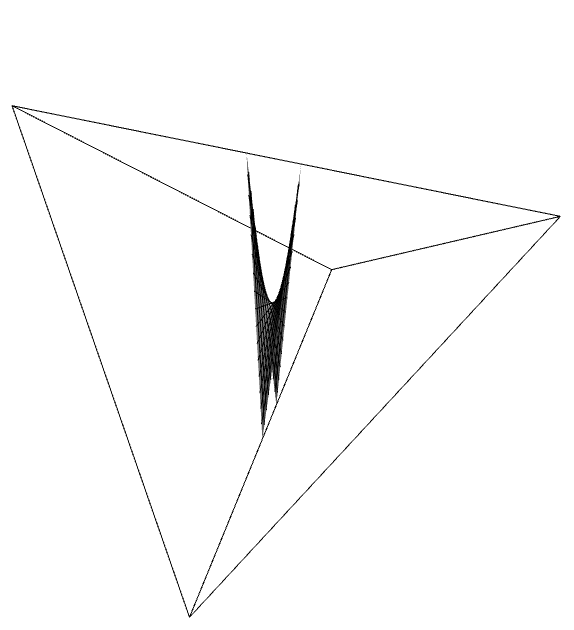}};
\node at (2,0) {\includegraphics[scale=0.34]{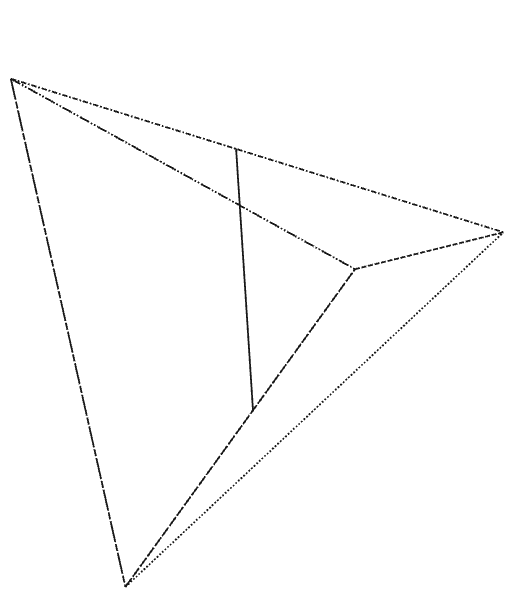}};
\node (X) at (0.8,-2) {$\delta_{01}$};
\node (X) at (0,0.9) {$\delta_{11}$};
\node (X) at (0,1.5) {$\delta_{00}$};
\node (X) at (-1.5,0.1) {$\delta_{10}$};
\node (X) at (-4.2,1.3) {$\delta_{00}$};
\node (X) at (-3,-2) {$\delta_{01}$};
\node (X) at (2.85,0) {$\delta_{10}$};
\node (X) at (3.95,0.75) {$\delta_{11}$};
\end{tikzpicture}
\end{flushleft}
\caption{\emph{Output distributions $q$ parameterised by the BSC parameters $b_1(0|0)$
and $b_2(0|0)$ for $p(1)=0.45$ (left) and $p(1)=1/2$ (right).
The geometric shape of the set of output distributions already hints at
the fact that it is not possible to retrieve $b_1(0|0)$ and $b_2(0|0)$ from $q$.
Cf.~\cite[Fig.~3]{ns16}.}}
\label{fig:tetra}
\end{figure}
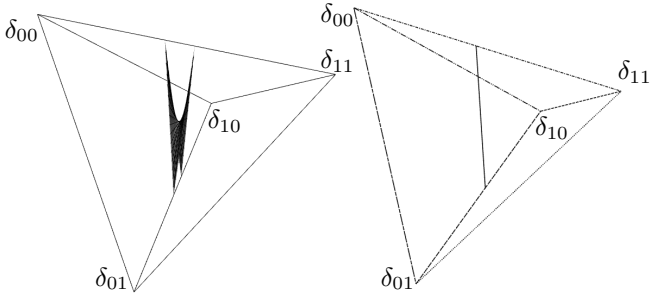

Fortunately, it is enough to go to $K=3$ and the following still holds.

\begin{theorem}
\label{theorem:main-BSC}
Let $K\geq3$ and $d$ be a distortion measure on $\cX\times\cX$.
Let $\mathbf{B}=(B_1,\ldots,B_K)$ be taken from the set $\{B:B$ is a $BSC$ with $b(0|0)\neq1/2\}^K$, and $p\in\mathcal{P}(\cX)$.
The uniform resolution of $\mathbf{W}$ with respect to $p$ is
\begin{align}
  \overline{\mathcal R}(p,\mathbf{W}) = d_{\mathrm{MCA}}(p,\mathbf W).
\end{align}
There is an algorithm that is uniformly universal on the set of
invertible binary symmetric
channels $\left\{B : B \text{ BSC with } b(0|0) \neq \frac12 \right\}$.
\end{theorem}

The intuition behind the proof of this result is inspired by the
following observation: If $K=3$ then the distribution of $X$ given
$X_1$ is not uniform anymore. Thus, conditioned on the outcome of
the first observed variable being e.g. a $0$, we can use the content
of Theorem \ref{theorem:DCA_BSC_K_2} to determine $b_2(0|0)$,
$b_3(0|0)$ and $B_1(p)$ (note that $B_1(p)=p$ holds).

When treating BSCs, it turns out that a universal algorithm
(called BUDDA in the following) may be defined based on analytical
formulas rather than outcomes of an optimisation procedure such as
in the definition of UDDA. The core of this approach is the content of the following theorem.

\begin{theorem}
\label{theorem:basis-for-UDDA}
Let $\cX=\{0,1\}$ and $K\geq2$. There exists a function
$f_K:\mathcal P(\cX^K)\to[0,1]$ such that the following holds:
If $p\in\mathcal P_>(\cX)$, $B_1,\ldots,B_K\in C(\cX,\cX)$ are
binary symmetric channels and $q\in\mathcal P(\cX^K)$ is the
output distribution of the dependent component system
$(p,B_1,\ldots,B_K)$ and $q_1,\ldots,q_K\in\mathcal P(\cX)$
denote the marginals of $q$ on the respective alphabets, then
\begin{align}
  p(0) &=\frac{1}{2}\left(1+\frac{|\prod_{i=1}^K(1-2\cdot q_i(0))|^{1/K}}{f_K(q)}\right).
\end{align}
The mapping $q\mapsto p(0)$ is henceforth denoted $E_K$.
\end{theorem}

The explicit nature of Theorem \ref{theorem:basis-for-UDDA}
(the definition of $f_K$ is the content of its proof)
allows us to develop error bounds in the determination of $p$,
when instead of the correct $q$ an estimate $\hat q$ is used for
the calculation. This in turn makes it possible to prove
Theorem \ref{theorem:main-BSC} based on the following definition of BUDDA.

\begin{algorithm}{BSC-UDDA}
To any pair $r,s\in\mathcal P(\{0,1\}$ such that
$r^\downarrow(0)<s^\downarrow(0)$, we let $b(r,s):=(s+r-1)/(2s-1)$
be the parameter such that a BSC $B$ with $b(0|0)=b(r,s)$ satisfies $B(s)=r$.

The following steps produce the desired estimate $\hat x^n$ for $x^n$ when $K\geq3$:
\begin{enumerate}
\item Input is $(y_{ij})_{i,j=1}^{n,K}$.
\item Calculate $\hat q(\cdot):=\tfrac{1}{n}N(\cdot|(\by_1,\ldots,\by_K)) \in \cP(\cX^K)$.
\item Calculate $\hat p=E_K(\hat q)$.
\item Let $i\in\{0,1\}$ be the number such that $N(y^n_1|i)>N(y^n_1|i\oplus1)$ where $\oplus$ is addition modulo two. Without loss of generality we assume in the following description that $i=0$. Calculate the conditional distribution $\hat q(\cdot|y_1=0)\in\mathcal P(\cX^{K-1})$.
\item If $\|\hat p-\pi\|>\min_i\|\hat q_i(\cdot|y_1=0)-\pi\|$, set $\hat b_i(0|0):=b(\hat q_i,\hat p)$.
\item If $\|\hat p-\pi\|<\min_i\|\hat q_i(\cdot|y_1=0)-\pi\|$, calculate $\hat p_1:=E_{K-1}(\hat q(\cdot|y_1=0)$ and $\hat p_2:=E_{K-1}(\hat q(\cdot|y_2=0)$. Then, set $\hat b_i(0|0):=b(\hat q_i(\cdot|y_1=0),\hat p_1)$ for $i=2,\ldots,K$ and $\hat b_1(0|0)=b(\hat q_i(\cdot|y_2=0),\hat p_2$.
\item Find the minimal clairvoyant ambiguous recovery $\mathcal T$ for $(\hat p,\hat B_1,\ldots,\hat B_K)$.
\item For each $i\in[n]$, apply the MCA decoder to the sequence $(y_{i1},\ldots,y_{iK})$.
      This generates a sequence $\hat x^n=\hat x_1 \ldots \hat x_n$.
\item Output $\hat x^n$.
\end{enumerate}
\end{algorithm}
As explained earlier, the fifth and sixth step in the algorithm allows one to reliably deal with those situations where $p\approx\pi$. Note that either of the two procedures may fail: the one in step five if $p\approx\pi$, the one in step six if $p(0)=1-a_1$. Luckily however, in the limit of $n\to\infty$ it is not possible that both conditions are fulfilled.
\begin{figure}[ht]
\begin{center}
\begin{tikzpicture}
\node (X) at (-4.2,6.5) {$X$};
\node (Y_1) at (-7,6) {$Y_1$};
\node (Y_2) at (-5.3,4.7) {$Y_2$};
\node (Y_3) at (-2,6) {$Y_3$};
\node (Xpr) at (-4.3,6.6) {};
\node (Y_1pr) at (-7,6.15) {};
\node (Y_2pr) at (-5.4,4.8) {};
\node (Y_3pr) at (-2,6.15) {};
\path[->, line width=1pt, shorten <= -1pt] (X) edge (Y_1);
\path[->, line width=1pt, shorten <= -1pt] (X) edge (Y_2);
\path[->, line width=1pt, shorten <= -1pt] (X) edge (Y_3);
\path[->, dashed, line width=1pt, shorten <= 4pt] (Y_1pr) edge (Xpr);
\path[->, dashed, line width=1pt, shorten <= +5pt, shorten >= +4pt] (Xpr) edge (Y_2pr);
\path[->, dashed,line width=1pt, shorten <= 7pt, shorten >= 3pt] (Xpr) edge (Y_3pr);
\end{tikzpicture}
\caption{\emph{Dependent component system with three outputs. Solid arrows denote information flow in the original model. Dashed arrows denote information flow when instead of the ''hidden variable`` $X$ the first visible alphabet $Y_1$ is treated as input to a dependent component system which then has only two outputs, $Y_2,Y_3$.}}
\end{center}
\end{figure}
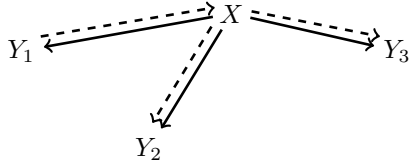
The fact that one output of the system is used to create a bias explains why at least three copies are needed.\\
This is especially important in the light of findings such as \cite{tomasi-bro2006} which demonstrate the instability of contemporary algorithms for fitting of tensor models.
\end{subsection}
\end{section}

%%%%%%%%%%%%%%%%%%%%%%%%%%%%%%%%%%%%%%%%%%%%%%%%%%%%%%%%%%%%%%%%%%%%%%%%%%%%%%%%%
%%%%%%%%%%%%%%%%%%%%%%%%%%%%%%%%%%%%%%%%%%%%%%%%%%%%%%%%%%%%%%%%%%%%%%%%%%%%%%%%%
\begin{section}{Simulation Results}
\label{sec:simulations}

Here we present simulations with a sample picture (a part of
Leonardo da Vinci's Mona Lisa, Fig.~\ref{fig:monalisa}), distorted
by ten different BSCs.

\begin{figure}[ht]
\begin{center}
\includegraphics[scale=0.2]{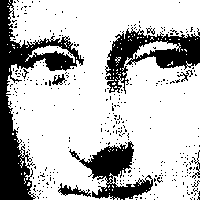}
\end{center}
\caption{\emph{The original image, a $200\times200$ pixels black-and-white picture
         of a part of Leonardo da Vinci's Mona Lisa.}}
\label{fig:monalisa}
\end{figure}

The parameters of the ten BSCs are chosen as
\begin{align*}
  (0.71, 0.32, 0.41, 0.49, 0.48, 0.82, 0.81, 0.51, 0.84, 0.17).
\end{align*}
The distorted images are shown in Fig.~\ref{fig:monalisa-noisy}.

\begin{figure}[ht]
\begin{center}
\begin{tabular}{lllll}
\includegraphics[scale=0.2]{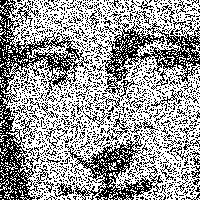}&\includegraphics[scale=0.2]{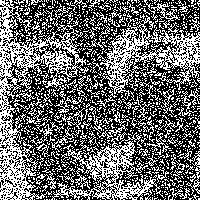}&\includegraphics[scale=0.2]{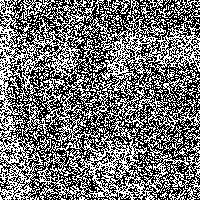}&\includegraphics[scale=0.2]{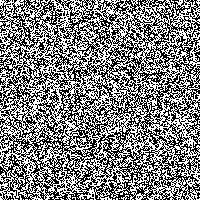}&\includegraphics[scale=0.2]{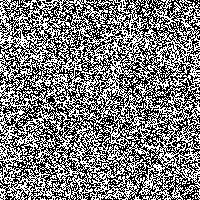}\\
\includegraphics[scale=0.2]{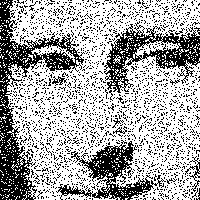}&\includegraphics[scale=0.2]{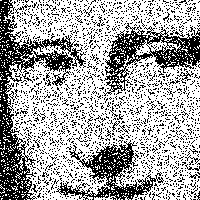}&\includegraphics[scale=0.2]{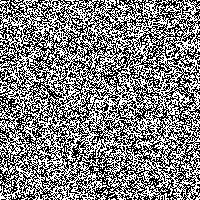}&\includegraphics[scale=0.2]{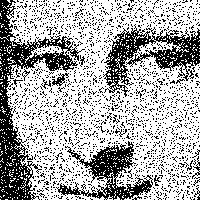}&\includegraphics[scale=0.2]{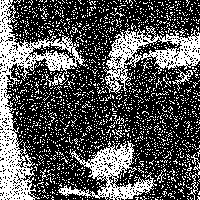}
\end{tabular}
\end{center}
\caption{\emph{The ten distorted versions of the original.}}
\label{fig:monalisa-noisy}
\end{figure}

Using $K=7$ if the noisy copies,
the channel parameters estimated by BUDDA are
\begin{align*}
(0.28, 0.67, 0.59, 0.51, 0.52, 0.17, 0.18),
\end{align*}
and the algorithm returns the denoised image shown in Fig.~\ref{fig:denoise-K7}.

\begin{figure}[ht]
\begin{center}
\begin{tabular}{ll}
\includegraphics[scale=0.2]{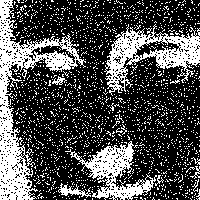}&\includegraphics[scale=0.2]{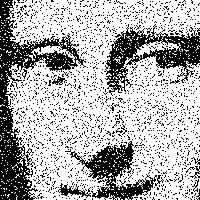}
\end{tabular}
\end{center}
\caption{\emph{$K=7$. Left: denoised image; right: denoised and inverted for easier comparison with the original.}}
\label{fig:denoise-K7}
\end{figure}

Now we proceed to use all our $K=10$ copies.
Then, BUDDA produces the following estimates for the channels:
\begin{align*}
  (0.29, 0.68, 0.59, 0.51, 0.52, 0.17, 0.19, 0.49, 0.16, 0.83),
\end{align*}
and the denoised image shown in Fig.~\ref{fig:denoise-K10}.

\begin{figure}[ht]
\begin{center}
\begin{tabular}{ll}
\includegraphics[scale=0.2]{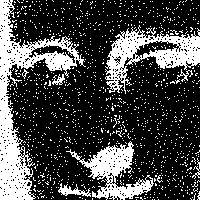}& \includegraphics[scale=0.2]{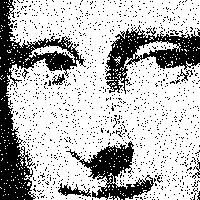}
\end{tabular}
\end{center}
\caption{\emph{$K=10$.
Left: denoised image; right: denoised and inverted for easier comparison with the original.}}
\label{fig:denoise-K10}
\end{figure}
\end{section}

%%%%%%%%%%%%%%%%%%%%%%%%%%%%%%%%%%%%%%%%%%%%%%%%%%%%%%%%%%%%%%%%%%%%%%%%%%%%%%%%%
%%%%%%%%%%%%%%%%%%%%%%%%%%%%%%%%%%%%%%%%%%%%%%%%%%%%%%%%%%%%%%%%%%%%%%%%%%%%%%%%
\begin{section}{Proofs}

\begin{proof}[Proof of Theorem \ref{theorem:main}]
%Let $n\in\nn$ and let, for sake of simplicity, $(x^n)_{n\in\nn}$ be a sequence of
%strings such that their associated sequence $(\overline{N}(\cdot|x^n))_{n\in\nn}$
%converges to $s\in\cP(\cX)$.
Let the empirical distribution of the output pixels produced in step 2 of the algorithm
be $\hat q\in\mathcal P(\cX^K)$ and $\eps>0$.
Application of the Chernoff-Hoeffding bound \cite{Hoeffding} yields
\begin{align}
  \mathbb P\left(\left\|\hat q - \left(\bigotimes_{j=1}^K W_j\right)p^{(K)} \right\| \leq \eps\right)
                                                                         \geq 1 - 2^{L^K} 2^{-2n\eps^2}.
\end{align}
This implies that step 2 of the algorithm produces a distribution which is, with arbitrarily high
accuracy and probability converging to $1$ as $n$ tends to infinity and
independently of the type of $x^n$, close to the distribution $\left(\bigotimes_{j=1}^K W_j\right)p^{(K)}$. Based on this estimate, we know
by
simple application of the triangle inequality that the function
$E:\mathcal P(\cX^K)\to\mathcal P^\downarrow(\cX^K)$ defined by
\begin{align}
  E(q) := \left(\bigotimes_{j=1}^K V_j\right)r^{(K)},
\end{align}
where (note that this next step is the implementation of the DCA algorithm)
\begin{align}
  (r,\mathbf V) = \operatorname*{arg\,min}_{(r^\downarrow,\mathbf V)}
                              \left\| \left( \bigotimes_{j=1}^K V_j \right)r^{\downarrow(K)}-q \right\|,
\end{align}
satisfies
\begin{align}
  \mathbb P\left(\left\| E(\hat q) - \left(\bigotimes_{j=1}^K W_j\right)p^{(K)} \right \|_1\leq 2\eps\right)
                                                                              \geq 1 - 2^{L^K} 2^{-2n\eps^2}.
\end{align}
As a consequence, we get
\begin{align}\label{eqn:above-observation}
  \mathbb P\left( \left\| E(\hat q) - \left(\bigotimes_{j=1}^K W_j\right)p^{(K)} \right \|_1
                                   \leq n^{-\eps/2}\right) \longrightarrow 1
\end{align}
as $n\to\infty$.
Take the function $\Theta''$ from Ref.~\cite[Theorem 1]{ns16},
which for $r\in\mathcal P_>(\cX)$ and channels $V_1,\ldots,V_K$
with invertible transition probability matrices is defined as
a one-sided inverse to $\Theta$:
\begin{align}
  \Theta''\left(\left(\bigotimes_{j=1}^K V_j\right)r^{(K)}\right) := (r^\downarrow,\mathbf V\circ\tau),
\end{align}
with $\tau\in S_\cX$ being a permutation with the property
$r(x)=r^\downarrow(\tau(x))$ for all $x\in\cX$, and $r^\downarrow \in \cP^\downarrow(\cX)$.
According to Lemma \ref{lemma:continuity-of-inverse} below,
$\Theta''$ is continuous if restricted to the set $P_>(\cX)$.
Based on the above observation, see in particular eq.~\eqref{eqn:above-observation},
$\Theta''$ satisfies
\begin{align*}
\mathbb P\bigl( \|\Theta''\circ E(\hat q)-(p,\mathbf W)\|_{DCS}
                      \leq\delta_{DCS}(p,\mathbf W,\eps_n) \bigr) \longrightarrow 1
\end{align*}
as $n\longrightarrow\infty$.
Thus, with probability approaching $1$ when $n$ tends to infinity and the image types
$\overline{N}(\cdot|x^n)$ are in a small enough vicinity of some distribution
$s'$, we know that our estimate of the system $(p,\mathbf W)$ will be accurate up to
an unknown permutation, and small deviations incurred by the properties of
$\delta_{DCS}(p,\mathbf W,\eps)$.
According to Lemma \ref{lemma:continuity-of-OA-recovery}, the minimum
clairvoyant ambiguous recovery computed based on the estimated system
$\Theta''\circ E(\hat q)$ is asymptotically optimal for the true system
$(p,\mathbf W)$. Using the Chernoff-Hoeffding inequality \cite{Hoeffding},
we can prove that application of the MCA decoder separately to each
letter $y_{i1},\ldots,y_{iK}$ yields a sequence $\hat x^n$ with the
property $n^{-1}\sum_{i=1}^n d(\tau(x_i),\hat x_i) < D$ with high probability
for every $D>d_{MCA}(p,\mathbf W)$.
Thus a look at \eqref{defn:distortion-measure} finishes the proof.
\end{proof}

We now state and prove the lemmas required above.

\begin{lemma}
\label{lemma:continuity-of-Theta}
Set $\Theta(s,\mathbf W) := \left(\bigotimes_{j=1}^K W_j\right)s^{(K)}$.
Then for every two dependent component systems $(r,\mathbf V)$ and
$(s,\mathbf W)$, we have
\begin{align}
  \left\|\Theta(r,\mathbf V)-\Theta(s,\mathbf W)\right\|
         \leq \left\|(r,\mathbf V)-(s,\mathbf W) \right\|_{DCS} \! .
\end{align}
That is, $\Theta$ is continuous with respect to $\|\cdot\|_{DCS}$.
\end{lemma}
\begin{proof}
We use a telescope sum identity: for arbitrary matrices
$A_1,\ldots,A_{K+1}$ and $B_1,\ldots,B_{K+1}$ we have
\begin{equation}\begin{split}
  \label{eqn:telescope-identity}
  A_1&\cdot\ldots\cdot A_{K+1}-B_1\cdot\ldots\cdot B_{K+1} \\
     &=\sum_{i=1}^{K+1}A_1\cdot\ldots\cdot A_{i-1}\cdot(A_i-B_i)B_{i+1}\cdot\ldots\cdot B_{K+1}.
\end{split}\end{equation}
We apply it to the channel matrices and input distributions as follows.
Let $A_i:=\eins^{\otimes (i-1)}\otimes V_i\otimes\eins^{\otimes(n-i)}$,
$B_i:=\eins^{\otimes (i-1)}\otimes W_i\otimes\eins^{\otimes(n-i)}$ for $i\in[K]$;
furthermore $A_{K+1}:=r^{(K)}$, $B_{K+1}:=s^{(K)}$, where we interpret
probability distributions on an alphabet as channels from a set with
one element.
Now we have
\[\begin{split}
  \bigl\|&(\otimes_{i=1}^KV_i)r^{(K)}-(\otimes_{i=1}^KW_i)s^{(K)}\bigr\| \\
         &=     \bigl\| A_1\cdot\ldots\cdot A_{K+1}-B_1\cdot\ldots\cdot B_{K+1} \bigr\| \\
         &=     \left\| \sum_{i=1}^{K+1} A_1\cdots A_{i-1}(A_i-B_i)B_{i+1}\cdots B_{K+1} \right\|_{FB} \\
         &\leq  \sum_{i=1}^{K+1} \bigl\| A_1\cdots A_{i-1}(A_i-B_i)B_{i+1}\cdots B_{K+1} \bigr\|_{FB} \\
         &=     \sum_{i=1}^{K+1} \bigl\| V_1\otimes\cdots\otimes V_{i-1}\otimes(V_i-W_i)
                                          \otimes W_{i+1}\otimes\cdots\otimes W_{K+1} \bigr\|_{FB} \\
         &=     \sum_{i=1}^{K+1} \bigl\| V_i-W_i \bigr\|_{FB} \\
         &=     \sum_{i=1}^{K} \bigl\| V_i-W_i \bigr\|_{FB} + \bigl\| r^{(K)}-s^{(K)}\bigr\| \\
         &=     \bigl\| (r,\mathbf V)-(s,\mathbf W) \bigr\|_{DCS},
\end{split}\]
concluding the proof.
\end{proof}

As a consequence, the restriction of $\Theta$ to
$\mathcal P^\downarrow(\cX)\times C(\cX,\cX)^K$,
which is invertible according to \cite[Theorem 1]{ns16}, is continuous as well.

% THE COROLLARY OF A LEMMA IS STILL A LEMMA:
%
\begin{lemma}
\label{lemma:continuity-of-inverse}
There exists a function
$\delta_{DCS}:\mathcal P^\downarrow(\cX)\times C(\cX,\cX)^K\times\mathbb R_{\geq 0}
\to\mathbb R_{\geq 0}$,
monotonic and continuous in the real argument and with $\delta_{DCS}(r,\mathbf W,0)=0$,
such that for every two dependent component systems $(s,\mathbf W)$ and $(r,\mathbf V)$
with $r,s \in \cP^\downarrow(\cX)$,
\begin{align}
  \left\| \left(\bigotimes_{j=1}^K W_j\right)s^{(K)}-\left(\bigotimes_{j=1}^K V_j\right)r^{(K)} \right\|
                                                \leq \eps
\end{align}
implies
\begin{align}
  \left\| (s,\mathbf W)-(r,\mathbf V) \right\|_{DCS} \leq \delta_{DCS}(r,\mathbf V,\eps).
\end{align}
\end{lemma}
\begin{proof}
Under the premise of the lemma, the function $\Theta$ is both
continuous on a connected domain
(as a consequence of Lemma \ref{lemma:continuity-of-Theta})
and invertible (as a consequence of \cite[Thm.~1]{ns16}).
Thus, the inverse of $\Theta$ is continuous, implying the claim.
\end{proof}

\begin{lemma}\label{lemma:continuity-of-OA-recovery}
Let $p,p'\in\mathcal P(\cX)$ and $C,C'\in\mathcal C(\cX,\cY)$.
Let $d$ be any distortion on $\cX\times\cX$ and $\mathcal T, \mathcal T'$ be the
MCA decoders for $(p,C)$ and $(p',C')$, respectively.
If $\|p-p'\|\leq\eps$ and $\|C-C'\|_{FB}\leq\eps$, then the minimal clairvoyant
ambiguous distortions satisfy
\begin{align}
  \min_{\tau\in S_\cX}\sum_{x,x'}&p'(x)d(x,x')c'(T_{\tau(x')}|x)\\
&\geq d_{\mathrm{MCA}}(p,C) - 2\epsilon|\cX||\cY| \max_{x,y}d(x,y).
\end{align}
\end{lemma}
\begin{proof}
The prerequisites of the lemma imply $|p(x)w(y|x)-p'(x)w'(y|x)|\leq\eps$
for all $x,y\in\mathcal X,\mathcal Y$. Thus
\begin{align}
\min_{\tau\in S_\cX}\sum_{x,x'}p'(x) &d(x,x')c'(T_{\tau(x')}|x) \\
              &\geq \min_{\tau\in S_\cX}\sum_{x,x'}p(x)d(x,x')c(T_{\tau(x')}|x)\nonumber\\
              &\phantom{=========}
                    - 2\epsilon|\cX\times\cY| \max_{x,y}d(x,y),
\end{align}
concluding the proof.
\end{proof}
\begin{proof}[Proof of Theorem \ref{theorem:basis-for-UDDA}]
Let $\mathcal X=\{0,1\}$,  $B_1,\ldots,B_K\in C_0(\mathcal X,\mathcal X)$ be BSCs and $p\in\mathcal P(\mathcal X)$. Let $q_i:=B_i(p)$. Then the following equation holds true:
\begin{align}
p(0)&=\frac{1}{2}\left(1+\prod_{i=1}^K\left|\frac{1-2\cdot q_i(0)}{1-2\cdot b_i}\right|^{1/K}\right),
\end{align}
where $b_i:=b_i(0|0)$, $i=1,\ldots,K$. It turns out that, for even $K$, the expression in the denominator can be derived as a function of $q=(B_1\otimes\ldots\otimes B_K)p^{(K)}$ in a very natural way. Set
\begin{align}\label{defn:DMN}
f_K(q):=\sum_{\tau\in S_2^K}\sgn(\tau)q(\tau_1(1),\ldots,\tau_K(1)),
\end{align}
where $S_2^K$ is the K-fold cartesian product of the set $S_2$ of permutations on $\{0,1\}$ with itself. Note that, given empirical data as a string $\by$, $f_K(\tfrac{1}{n}N(\cdot|\by))$ can be calculated for even $K$ as
\begin{align}\label{defn:empirical-DMN}
f_K\left(\tfrac{1}{n}N(\cdot|\by)\right)&=\sum_{i=1}^n(-1)^{\sum_{j=1}^Ky_{ij}}.
\end{align}
This formula may be particularly useful for calculating $f_K$ in an online fashion, while data is being retrieved from a source. Going further, we see that it holds (with $p:=p(0)$ and $p':=p(1)$:
\begin{align}
f_K(q)&=\sum_{\tau\in S_2^K}\sgn(\tau)p\cdot\prod_{i=1}^Kb_i(\tau_i(0)|0)\nonumber\\
&\qquad+\sum_{\tau\in S_2^K}\sgn(\tau)p'\cdot\prod_{i=1}^Kb_i(\tau_i(0)|1)\\
&=p\sum_{\tau\in S_2^K}\sgn(\tau)\prod_{i=1}^Kb_i(\tau_i(0)|0)\nonumber\\
&\qquad+p'\sum_{\tau\in S_2^K}\sgn(\tau)\prod_{i=1}^Kb_i(\tau_i(0)|1)\\
&=p\sum_{\tau\in S_2^K}\sgn(\tau)\prod_{i=1}^Kb_i(\tau_i(0)|0)\nonumber\\
&\ +p'\sum_{\tau\in S_2^K}\sgn(\tau\circ\mathbb F^{\otimes 4})\prod_{i=1}^Kb_i(\tau_i\circ\mathbb F(0)|1)\\
&=p\sum_{\tau\in S_2^K}\sgn(\tau)\prod_{i=1}^Kb_i(\tau_i(0)|0)\nonumber\\
&\qquad+p'\sum_{\tau\in S_2^K}\sgn(\tau)\prod_{i=1}^Kb_i(\tau_i(1)|1)\\
&=p\sum_{\tau\in S_2^K}\sgn(\tau)\prod_{i=1}^Kb_i(\tau_i(0)|0)\nonumber\\
&\qquad+p'\sum_{\tau\in S_2^K}\sgn(\tau)\prod_{i=1}^Kb_i(\tau_i(0)|0)\\
&=\sum_{\tau\in S_2^K}\sgn(\tau)\prod_{i=1}^Kb_i(\tau_i(0)|0),
\end{align}
so it does not depend on the distribution $p$ at all. Going further, we see that
\begin{align}
f_K(q)&=\sum_{\tau\in S_2^K}\sgn(\tau)\prod_{i=1}^Ka_i(\tau_i(0)|0)\\
&=\prod_{i=1}^K\sum_{\tau\in S_2}\sgn(\tau)a_i(\tau(0)|0)\\
&=\prod_{i=1}^K(a_i-1+a_i)\\
&=\prod_{i=1}^K(1-2 a_i).
\end{align}
For uneven $K$, $K\geq3$, the above calculations demonstrate one can define $f_K$ by
\begin{align}
f_K(q):=\left(\prod_{i=1}^Kf_{K-1}(q_{\neq i})\right)^{1/K}
\end{align}
where $q_{\neq i}$ is the marginal of $q$ on all but the $i$-th alphabet. This proves the Theorem.
\end{proof}
\begin{proof}[Proof of Theorem \ref{theorem:main-BSC}]
In the remainder, it is assumed that $p$ is the (empirical) distribution of the sought after data string $x^n\in\{0,1\}^n$. This implies that probabilities are calculated with respect to the distribution on the $K$-fold product $\{0,1\}^K$ that is defined as the marginal on $\cY^K$ of
\begin{align}\label{eqn:probabilitstic-law-in-proof}
\mathbb P(Y=y_1,\ldots,y_K,X=x)=\prod_{j=1}^{K}b_j(y_{j}|x)\frac{1}{n}N(x|x^n).
\end{align}
Since it simplifies notation, all following steps are however performed as if $p\in\cP(\{0,1\})$ was an arbitrary (meaning that it is not necessarily an empirical distribution) parameter of the system.

To every such parameter $p$, there is a set of associated distributions $p_X(\cdot|Y_i=0)$, $i=1,\ldots,K$, which are calculated from \eqref{eqn:probabilitstic-law-in-proof} as $p_X(X=x|Y_i=0)=\mathbb P(X=x,Y_i=0)\cdot\mathbb P(Y_i=0)^{-1}$ if $\mathbb P(Y_i=0)>0$, and $p_X(X=x|Y_i=0)=\delta_1(x)$ otherwise.
Since BUDDA always assumes that $N(0|x^n)\geq N(1|x^n)$, this latter choice is justified.

By definition of $E_K$ in Theorem \ref{theorem:basis-for-UDDA}, it is clear that
$E_K$ is continuous. Thus, from the Chernoff-Hoeffding bound (see e.g.~\cite[Thm.~1.1]{DP})
we have that for every $\eps>0$ there is an $N_1=N_1(\eps,B_1,\ldots,B_K)\in\nn$ such
that, if $n\geq N_1$, the probability of the set $S_0:=\{y^n:\|E_K(q)-p^\downarrow\|_1\leq\eps\}$
satisfies
\begin{align}
  \mathbb P(S_0) \geq 1-\nu(n)
\end{align}
for some function $\nu:\nn\to(0,1)$ of the form $\nu(n)=2^{-n c}$ for some constant $c>0$
depending only on the channels $B_1,\ldots,B_K$. Obviously, $\nu$ converges to $0$
when $n$ goes to infinity. The fact that this convergence is exponentially
fast makes it possible to prove that a larger number of events has asymptotically
the probability one by repeated union bounds. Setting
\begin{align}
S_i:=\{y^n:\|E_{K-1}(q(\cdot|y_i=0))-p_X(\cdot|y_i=0)\|_1\leq\eps\},
\end{align}
we have then
\begin{align}
\mathbb P(\cap_{i=0}^KS_i)\geq 1-(K+1)\nu(n).
\end{align}
It is critical to note here that we use one and the same function $\nu$ over and over again. Theoretically, we would need to update the constant $c$ at each use. We do however step away from such explicit calculation, since any positive $c'<c$ that would arise e.g. from setting $c':=\min\{c_1,\ldots,c_K\}$ is sufficient for our proof. Thus we use one and the same constant repeatedly in what follows.

As an outcome of the calculation of $p$ via application of $E_K$ to $q$, one is in principle able to obtain an estimate of the channels $B_1,\ldots,B_K$. However, the precision of this estimate depends critically on the distance $\|p-\pi\|$ and thus any algorithm using only such estimate cannot be uniformly universal because the equality $p=\pi$ may hold. Define
\begin{align}
E:=\{y^n:\|E_K(\hat q)-\pi\|>\min_i\|E_{K-1}(\hat q(\cdot|y_i=0))-\pi\|\}
\end{align}
and
\begin{align}
F:=\{y^n:\|E_K(\hat q)-\pi\|\leq\min_i\|E_{K-1}(\hat q(\cdot|y_i=0))-\pi\|\}.
\end{align}
Observe that the distribution of $X$ given $X_1=0$ satisfies
\begin{align}\label{eqn:conditioned-system}
\mathbb P(X=0|Y_1=0)&=\frac{p(0)\cdot b_1}{p(0)\cdot b_1+(1-p(0))(1-b_1)}.
\end{align}
The left hand side in \eqref{eqn:conditioned-system} takes on the value $b_1\neq1/2$ at $p(0)=1/2$, and the value $1/2$ at $p(0)=1-b_1$. Since by assumption we have that $b_i\neq1/2$ for all $i=1,\ldots,K$ it is clear that, asymptotically, either $S_1$ or $S_2$ has large probability. Thus, depending on the distribution $p$, we have either
\begin{align}
\mathbb P(E) \geq 1-\nu(n)
\end{align}
or else
\begin{align}
\mathbb P(F) \geq 1-\nu(n),
\end{align}
but always $\mathbb P(E\cap F)<\nu(n)$. This reasoning forms the basis of step $5)$ of BUDDA. If $y^n\in S_2$, the proof of Theorem \ref{theorem:main} shows that BUDDA succeeds with probability lower bounded by
$1-2^{-nc'}$ for some $c'>0$ depending only on the channels $B_1,\ldots,B_K$. If $y^n\in S_3$, then one introduces a bias by viewing the outcome $Y_1$ as the source of the system, and then conditioning on the outcomes of $Y_1$: Delete all rows in $y^n=(y_{i,j})_{i,j=1}^{K,n}$ with indices $j$ such that $y_{1,j}=1$. The resulting empirical distribution $\hat q(\cdot|y_1=0)$ serves as input to $E_{K-1}$.
It produces a dependent component system $(\hat b_2,\ldots,\hat b_K,\hat p_0)$. Performing the same step with $k$ set to two yields a second such system with an additional value for $\hat b_1$.

Thus steps $1)$ to $6)$ of BUDDA succeed with probability lower bounded by $1-(K+4)\nu(n)$,
\emph{independently} of the sequence $x^n$. The remaining steps of the proof are identical
to those in the proof of Theorem \ref{theorem:main}. They are based on further application
of the union bound and the fact that $\lim_{n\to\infty} M\cdot 2^{-nc}=0$ for every $M\in\nn$.
\end{proof}

\begin{remark}
A more optimized choice of estimator arises when one picks instead of the
first and the second system those two systems for which $\|E_{K-1}(\hat q(\cdot|y_i=0))-\pi\|$
takes on the largest and second largest value, respectively.
\end{remark}

\end{section}

%%%%%%%%%%%%%%%%%%%%%%%%%%%%%%%%%%%%%%%%%%%%%%%%%%%%%%%%%%%%%%%%%%%%%%%%%%%%%%%%
%%%%%%%%%%%%%%%%%%%%%%%%%%%%%%%%%%%%%%%%%%%%%%%%%%%%%%%%%%%%%%%%%%%%%%%%%%%%%%%%
\begin{section}{Conclusions}
We have laid the information theoretic foundations of universal discrete denoising and
demonstrated the in-principle feasibility of our new algorithm UDDA based on
discrete dependent component analysis.
We proved that the resolution equals the minimal clairvoyant ambiguous distortion.
We conjecture that, in fact, the same distortion can be achieved uniformly in
the input distribution, i.e.~that the equality
$\overline{\mathcal{R}}(p,\mathbf W)=d_{MCA}(p,\mathbf W)$ holds,
but have only been able to prove this for a restricted class of models and
with some slight modifications to UDDA.
\end{section}

%%%%%%%%%%%%%%%%%%%%%%%%%%%%%%%%%%%%%%%%%%%%%%%%%%%%%%%%%%%%%%%%%%%%%%%%%%%%%%%%%
%%%%%%%%%%%%%%%%%%%%%%%%%%%%%%%%%%%%%%%%%%%%%%%%%%%%%%%%%%%%%%%%%%%%%%%%%%%%%%%%%

\section*{Appendix A: Additional Proofs}
\label{sec:appendix-A}

\begin{proof}[Proof of eq.~(\ref{eqn:BSC-problemfor-K=2})]
Let $A,B\in C_0(\{0,1\},\{0,1\})$ be BSCs and $p=\pi$. Then the following holds true:
\begin{align}
(A\otimes B)p^{(2)} &= \frac12 (A\otimes B)\sum_{i=0}^1e_i^{\otimes 2}\\
                    &= \frac12 (A\circ B^{T}\otimes\eins)\sum_{i=0}^1e_i^{\otimes 2}\\
                    &= \frac12 (C\otimes\eins)\sum_{i=0}^1e_i^{\otimes 2}
\end{align}
where $C=A\circ B^T$. Since $B^T=B$ for a BSC, this implies $C=A\circ B$, and especially that
\begin{align}
c &= ab+(1-a)(1-b)\\
  &= 2ab+1-a-b,
\end{align}
where we used the abbreviation $c:=c(0|0)$, $b:=b(0|0)$ and $a:=a(0|0)$.
\end{proof}

\begin{lemma}
Let $W\in\mathcal C(\cX,\cY)$. Given $x\in\mathcal X$,
let $\mathcal Y_x$ denote the random variable with distribution
$w(\cdot|x)$. Let $\mathcal T$ be an estimator achieving $\Lambda=\sum_{x}p(x)w(T_x|x)$.
If $x^n$ has type $p$, then for every possible input string $x^n$ the
letter-by-letter application of $\mathcal T$ to the output sequences of $W^{\otimes n}$
given input $x^n$ has the following property:
\begin{align}
  \mathbb P\left(\left|\frac{1}{n}\sum_{i=1}^n\delta(x_i,\mathcal T(Y_{x_i}))-\Lambda\right|\geq\eps\right)
          \leq 2\cdot 2^{-2\eps^2 n}.
\end{align}
\end{lemma}
\begin{proof}
The expectation of the function $\Phi$ defined by
$\Phi(y^n):=\frac{1}{n}\sum_{i=1}^n\delta(x_i,\mathcal T(Y_{x_i}))$ is calculated as
\begin{equation}
  \mathbb E\Phi = \sum_xp(x)\sum_{y}\delta(x,y)w(\mathcal T(y)|x) = \Lambda.
\end{equation}
Thus by Hoeffding's inequality \cite[Thm.~1.1]{DP} or alternatively
McDiarmid's inequality \cite{mcdiarmid},
\begin{align}
  \mathbb P\left(\left|\frac{1}{n}\sum_{i=1}^n\delta(x_i,\mathcal T(Y_{x_i}))-\Lambda\right|\geq\eps\right)
           \leq 2\cdot 2^{-2\eps^2 n}.
\end{align}
\end{proof}

\begin{remark}
As a consequence of this lemma, an algorithm that knows $p$ and $W$ but not $x^n$
can calculate the MCA decoder $\mathcal T$ corresponding to $p$ and $W$, after
which it applies the MCA decoder in order to obtain the best possible estimate on $x^n$.
\end{remark}

%%%%%%%%%%%%%%%%%%%%%%%%%%%%%%%%%%%%%%%%%%%%%%%%%%%%%%%%%%%%%%%%%%%%%%%%%%%%%%%%
%%%%%%%%%%%%%%%%%%%%%%%%%%%%%%%%%%%%%%%%%%%%%%%%%%%%%%%%%%%%%%%%%%%%%%%%%%%%%%%%
\bigskip\noindent
\textit{Acknowledgments.}
The present research was funded by the German Research Council (DFG), grant no.~1129/1-1,
the BMWi and ESF, grant 03EFHSN102 (JN), the ERC Advanced Grant IRQUAT,
the Spanish MINECO, projects FIS2013-40627-P and FIS2016-86681-P, with the
support of FEDER funds, and the Generalitat de Catalunya, CIRIT project
2014-SGR-966 (AW and JN).

%\vfill

\bibliographystyle{IEEEtran}
%\bibliography{../../../percolation}
%\bibliography{percolation}

\end{document}